\documentclass[letterpaper, 10 pt, conference]{ieeeconf}  
\usepackage[left=0.75in, right=0.75in, top=0.75in, bottom=0.8in]{geometry}

\IEEEoverridecommandlockouts                              

\usepackage{acronym}
\usepackage{algorithm}
\usepackage{algpseudocode}
\usepackage{amsmath} 
\usepackage{amssymb}  
\usepackage{anyfontsize}
\usepackage{arydshln}

\usepackage{bm}

\usepackage{cases}
\usepackage{cite}
\usepackage{color}
\usepackage{csquotes}
\usepackage{caption}

\usepackage{epsfig} 
\usepackage{etoolbox}
\usepackage{eucal}

\usepackage{graphics} 
\usepackage{graphicx}
\usepackage{grffile}

\usepackage{hyperref}

\usepackage{indentfirst}

\usepackage{lipsum}
\usepackage{listings}

\usepackage{mathptmx} 
\usepackage{mathrsfs}
\usepackage{mathtools}

\usepackage{pgfplots}

\usepackage{relsize}

\usepackage{soul}
\usepackage{subcaption}

\usepackage{times} 
\usepackage{tikz}

\usepackage[nameinlink]{cleveref} 

\pgfplotsset{compat=1.16}

\creflabelformat{equation}{#2#1#3}
\crefname{equation}{Equation}{Equations}

\crefname{section}{Section}{Sections}

\crefname{table}{Table}{Tables}
\crefname{figure}{Figure}{Figures}
\crefname{algorithm}{Algorithm}{Algorithms}

\Crefname{table}{Table}{Tables}
\Crefname{figure}{Figure}{Figures}
\Crefname{algorithm}{Algorithm}{Algorithms}

\crefname{ineq}{Inequality}{Inequalities}
\Crefname{ineq}{Inequality}{Inequalities}
\creflabelformat{ineq}{#2{\upshape#1}#3}

\crefname{prob}{Problem}{Problems}
\Crefname{prob}{Problem}{Problems}
\creflabelformat{prob}{#2{\upshape#1}#3} 

\crefname{lem}{Lemma}{Lemmas}
\Crefname{lem}{Lemma}{Lemmas}
\creflabelformat{lem}{#2{\upshape#1}#3}

\usetikzlibrary{patterns}

\newtheorem{definition}{Definition}

\newtheorem{theorem}{Theorem}
\newtheorem{lemma}[theorem]{Lemma}

\DeclarePairedDelimiter\set{\{}{\}}

\newcommand{\diag}{\text{diag}}

\newcommand{\cC}{\mathcal{C}}

\newcommand{\cE}{\mathcal{E}}

\newcommand{\cG}{\mathcal{G}}
\newcommand{\cH}{\mathcal{H}}

\newcommand{\cL}{\mathcal{L}}

\newcommand{\cO}{\mathcal{O}}

\newcommand{\cQ}{\mathcal{Q}}

\newcommand{\cV}{\mathcal{V}}

\newcommand{\cZ}{\mathcal{Z}}

\newcommand{\bbA}{\mathbb{A}}

\newcommand{\bbN}{\mathbb{N}}

\newcommand{\bbP}{\mathbb{P}}
\newcommand{\bbQ}{\mathbb{Q}}
\newcommand{\bbR}{\mathbb{R}}
\newcommand{\bbS}{\mathbb{S}}

\newcommand{\bbX}{\mathbb{X}}

\newcommand{\bzero}{\mathbf{0}}

\newcommand{\bA}{\mathbf{A}}
\newcommand{\bB}{\mathbf{B}}
\newcommand{\bC}{\mathbf{C}}
\newcommand{\bD}{\mathbf{D}}
\newcommand{\bE}{\mathbf{E}}

\newcommand{\bG}{\mathbf{G}}
\newcommand{\bH}{\mathbf{H}}
\newcommand{\bI}{\mathbf{I}}
\newcommand{\bJ}{\mathbf{J}}
\newcommand{\bK}{\mathbf{K}}

\newcommand{\bP}{\mathbf{P}}
\newcommand{\bQ}{\mathbf{Q}}
\newcommand{\bR}{\mathbf{R}}
\newcommand{\bS}{\mathbf{S}}
\newcommand{\bT}{\mathbf{T}}
\newcommand{\bU}{\mathbf{U}}
\newcommand{\bV}{\mathbf{V}}
\newcommand{\bW}{\mathbf{W}}
\newcommand{\bX}{\mathbf{X}}
\newcommand{\bY}{\mathbf{Y}}
\newcommand{\bZ}{\mathbf{Z}}

\newcommand{\be}{\mathbf{e}}

\newcommand{\bu}{\mathbf{u}}

\newcommand{\bx}{\mathbf{x}}
\newcommand{\by}{\mathbf{y}}

\newcommand{\barbJ}{\overline{\bJ}}

\newcommand{\barbQ}{\overline{\bQ}}

\newcommand{\barbX}{\overline{\bX}}
\newcommand{\barbY}{\overline{\bY}}

\newcommand{\barbbQ}{\overline{\bbQ}}

\newcommand{\barcQ}{\overline{\cQ}}

\newcommand{\hatbJ}{\widehat{\bJ}}

\newcommand{\hatbQ}{\widehat{\bQ}}

\newcommand{\hatbX}{\widehat{\bX}}

\newcommand{\hatbbQ}{\widehat{\bbQ}}

\newcommand{\tilbJ}{\widetilde{\bJ}}

\newcommand{\tilbbQ}{\widetilde{\bbQ}}

\newcommand{\thh}{\mathrm{th}}
\usepackage[acronym]{glossaries}

\makeglossaries

\newacronym{vndt}{VNDT}{Vidyasagar's Network Dissipativity Theorem}
\newacronym{kyp}{KYP Lemma}{Kalman-Yakubovixh-Popov Lemma}
\newacronym{io}{IO}{input-output}

\newacronym{lti}{LTI}{linear time-invariant}
\newacronym{lmi}{LMI}{linear matrix inequality}

\newacronym{admm}{ADMM}{alternating direction methods of multipliers}
\newacronym{sdp}{SDP}{semidefinite programming}

\newacronym{uav}{UAV}{unmanned aerial vehicle}
\title{\LARGE \bf
Consensus-Based Stability Analysis of Multi-Agent Networks
}

\author{Ingyu Jang$^{1}$, Ethan J. LoCicero$^{2}$ and Leila Bridgeman$^{1}$
\thanks{*This work is supported by ONR Grant No. N00014-23-1-2043.}
\thanks{$^{1}$Ingyu Jang (PhD Student), and Leila Bridgeman (Assistant Professor) are with the Department of Mechanical Engineering and Material Science, Duke University, Durham, NC 27708 USA
        (email: {\tt\small ij40@duke.edu; ljb48@duke.edu}, phone: 919-225-4215)}%
\thanks{$^{2}$Ethan J. LoCicero (Fellow) is with ARPA-E, Washington, DC 20585 USA
        (email: {\tt\small ethanlocicero@gmail.com})}%
}

\begin{document}

\maketitle
\thispagestyle{empty}
\pagestyle{empty}

\begin{abstract}

The emergence of large-scale multi-agent systems has led to controller synthesis methods for sparse communication between agents. However, most sparse controller synthesis algorithms remain centralized, requiring information exchange and high computational costs. This underscores the need for distributed algorithms that design controllers using only local dynamics information from each agent. This paper presents a consensus-based distributed stability analysis. The proposed stability analysis algorithms leverage Vidyasagar's Network Dissipativity Theorem and the alternating direction methods of multipliers to perform general stability analysis. Numerical examples involving a 2D swarm of unmanned aerial vehicles demonstrate the convergence of the proposed algorithms.

\end{abstract}

\section{INTRODUCTION}

Recently, large-scale multi-agent networked systems have emerged in practical fields, driven by rapid advancements in digital communication and sensing technologies \cite{sztipanovits2011toward}. In response to these trends, controller synthesis methods have been developed to design sparse communication between agents, resolving scalability challenges \cite{jovanovic2016controller}. However, despite the sparse interconnections, the stability analysis and controller synthesis of such sparsity-promoting algorithms remain centralized. Centralized approaches are often impractical due to intellectual property concerns \cite{isik2023impact} and the high computational costs related to larger networks. This paper achieves consensus-based distributed stability analysis, where agents share only coarse dissipativity information, using decentralized methods to verify \gls{vndt} \cite{vidyasagar1981input}. This work provides a foundation for future work in distributed sparse controller network synthesis for nonlinear systems.

Dissipativity \cite{willems1972dissipative,hill1977stability} is a more general \gls{io} stability result than many \gls{io} theorems, such as the Passivity \cite{hatanaka2015passivity}, Small Gain \cite{zames1966input1}, and Conic Sector Theorems \cite{zames1966input1}. These \gls{io} methods model systems as input-output mappings, without consideration for the internal mechanisms in between. \gls{vndt} incorporates the dissipative properties of individual agents to analyze the stability of multi-agent systems. However, using this theorem directly still results in a centralized stability analysis, so further work is needed to avert the limitations of centralized algorithms.

Distributed optimization principles can be applied to \gls{io} stability, enabling localized system analysis. While extensively studied in parallel and distributed computation \cite{bertsekas2015parallel,tsitsiklis1986distributed}, distributed optimization has attracted interest in networked systems due to demands such as security and computational efficiency \cite{nedich2015convergence,nedic2018network}. A recent survey \cite{yang2019survey} highlights significant progress in distributed algorithms, but most of the research addresses unconstrained problems, which precludes directly imposing constraints that ensure network-wide stability. Although some algorithms address global, network-wide, constraints \cite{liu2017constrained}, they rely on specific network structures, such as undirected graphs. The \gls{admm}, in contrast, is well-suited for distributed algorithms with global constraints \cite{boyd2011distributed}.

In \cite{agarwal2020distributed}, a distributed dissipativity analysis was proposed based on an iterative Schur complement, but the approach is limited to \gls{lti} agents due to its reliance on the \gls{kyp} \cite{gupta1996robust}. In addition, the stability analysis in this approach assumes fixed dissipative properties, which may limit its applicability. To address these limitations, this work introduces a more general framework which can be modified for the stability analysis of systems with nonlinear agents such as \cite{li2002delay, mahmoud2009dissipativity, haddad2022dissipativity}, while enabling the use of flexible dissipativity parameters beyond conventional passivity and $\cL_2$-gain. In the proposed method, each agent independently analyzes its own dissipativity, shares its identified dissipativity parameters with its neighbors, and iteratively optimizes its dissipativity parameters to ensure the stability of the network.

\section{PRELIMINARIES} \label{chap:Preliminaries}

\subsection{Notation}
The sets of real numbers, natural numbers, and natural numbers up to $n$ are denoted $\bbR$, $\bbN$, $\bbN_n$, respectively. Bold uppercase letters, $\bA$, represent matrices, lowercase, $\bx$, represent vectors, and simple letters, $x$, represent scalars. The notation $\bA\prec0$ indicates that $\bA$ is negative-definite. The set of real $n\times m$ matrices is $\bbR^{n\times m}$, and the $(i,j)^\thh$ block or element of a matrix $\bA$ is denoted $(\bA)_{ij}$. If $(\bA)_{ij}\in\bbR^{n_i\times m_j}$ and $\bA\in\bbR^{\sum_{i=1}^N n_i \times \sum_{j=1}^M m_j}$, then $(\bA)_{ij}$ is said to be a ``block" of $\bA$, and $\bA$ is said to be in $\bbR^{N\times M}$ block-wise. The block diagonal matrix of $\bA_i$ for all $i\in\bbX$ is denoted $\diag(\bA_i)_{i\in{\bbX}}$. The cardinality of a set $\bbA$ is denoted by $|\bbA|$. The set of $n\times n$ symmetric matrices is denoted by $\bbS^n$. The $n\times n$ identity matrix is denoted by $\bI_n$. The set of square integrable functions is $\cL_{2}$. The Frobenius norm and $\cL_2$ norm are denoted by $\|\cdot\|_F$ and $\|\cdot\|_2$, respectively.  The truncation of a function $\by(t)$ at $T$ is denoted by $\by_T(t)$, where $\by_T(t)=\by(t)$ if $t\leq T$, and $\by_T(t)=0$ otherwise. If $\|\by_T\|_2^2{=}\langle\by_T,\by_T\rangle{=}\int_0^{\infty}\by_T^T(t)\by_T(t)dt{<}\infty$ for all $T{\geq}0$, then $\by{\in}\cL_{2e}$, where $\cL_{2e}$ is the extended $\cL_2$ space. The indicator function is denoted by $I_\bbA:\Omega\to\set{0,1}$,where $I_\bbA(x)=0$ if $x\in\bbA$ and $I_\bbA(x)=1$ otherwise for all $x\in\Omega$.

\subsection{Graph Structure and Chordal Decomposition}
A graph, $\cG(\cV(\cG),\cE(\cG))$, is defined by its vertices, $\cV(\cG)=\bbN_N$, and edges, $\cE(\cG)\subseteq\cV(\cG)\times\cV(\cG)$. It is undirected if $(i,j)\in\cE(\cG)$ implies $(j,i)\in\cE(\cG)$; otherwise, it is directed. A clique, $C\subseteq\cV$, is a set of $i,j\in C$ satisfying $i\neq j$, $(i,j)\in\cE(\cG)$. It is maximal if it is not a subset of another clique. A cycle of length $\alpha$ is a set of pairwise distinct vertices $\bbN_{\alpha}\subseteq\cV(\cG)$ such that $(\alpha,1)\in\cE(\cG)$ and $(i,i+1)\in\cE(\cG)$ for $i\in\bbN_{\alpha-1}$. A chord is an edge connecting non-consecutive vertices within a cycle. A chordal graph is an undirected graph in which every cycle of length greater than three contains a chord.

Set operations can be applied to graphs \cite{diestel2024graph}. For two graphs $\cG$ and $\cG'$, $\cG\cap\cG'=\big(\cV(\cG)\cap\cV(\cG'),\cE(\cG)\cap\cE(\cG')\big)$. If $\cG\cap\cG'=\varnothing$, then $\cG$ and $\cG'$ are disjoint. If $\cV(\cG')\subseteq\cV(\cG)$ and $\cE(\cG')\subseteq\cE(\cG)$, then $\cG'$ is a subgraph of $\cG$.

Graphs are used throughout this work to represent the structure of matrices. When this is done, both will share the same letter, but the graph will be in script letters, e.g. graph $\cG$ indicates the (block-wise) structure of matrix $\bG$, meaning the element (block) $(\bG)_{ij} \neq \textbf{0}$ if and only if $(i,j)\in\cE(\cG)$. Let $\bbS_-^n(\cE(\cG),0){=}\set*{\bG\in\bbS^n\,|\,\bG\preceq0,(\bG)_{ij}{=}\textbf{0}\text{ if }(i,j)\notin\cE(\cG)}$ be the set of negative semi-definite matrices structured according to $\cG$. The following theorem is useful for decomposing negative semi-definite block matrices.
\begin{theorem} [Chordal Block-Decomposition \cite{zheng2021chordal}] \label{thm:Chordal Decomposition}
    Let $\cZ$ be a chordal graph with maximal cliques $\set{C_i}_{i=1}^M$. Then $\bZ\in\bbS_{-}^N(\cE(\cZ),0)$ (block-wise) if and only if there exist $\bZ_p\in\bbS_-^{|C_p|}$ (block-wise) for $p\in\bbN_M$ such that
    \small
    \begin{align} \label{eq:Chordal Decomposition}
        \bZ = \sum_{p=1}^M\bE_{C_p}^T\bZ_p\bE_{C_p},
    \end{align}
    \normalsize
    where $\bE_{C_p}\in\bbR^{|C_p|\times N}$ (block-wise) is defined as $(\bE_{C_p})_{ij}=\bI$ if $C_p(i)=j$ and $(\bE_{C_p})_{ij}=\textbf{0}$ otherwise, 
    and $C_p(i)$ is the $i^\thh$ vertex of $C_p$.
\end{theorem}

\subsection{\texorpdfstring{\gls{admm}}{ADMM}} \label{chap:ADMM}

\gls{admm} is an algorithm that can be used for decentralized optimization\cite{boyd2011distributed}. Consider the problem,
\begin{align} \label[prob]{eq:Constrained Optimization Problem}
\begin{split}
    \min_{\bX}\quad&f(\bX) \qquad
    \text{s.t.}\quad\bX\in\Omega,
\end{split}
\end{align}
where $\bX\in\bbR^{n\times m}$ is the ``primal" variable, $f:\bbR^{n\times m}\to \bbR\cup\set{+\infty}$ is the convex objective function, and $\Omega\subseteq\bbR^{n\times m}$ is the convex constraint set. 
\cref{eq:Constrained Optimization Problem} can be reformulated as
\begin{align} \label[prob]{eq:Equivalent Constrained Optimization Problem}
    \begin{split}
        \min_{\bX,\bZ}\quad&f(\bX)+I_\Omega(\bZ) \qquad
        \text{s.t.}\quad\bX-\bZ=\bzero,
    \end{split}
\end{align}
with ``clone" variable, $\bZ\in\bbR^{n\times m}$. \gls{admm} iteratively solves \cref{eq:Equivalent Constrained Optimization Problem} by calculating
\begin{subequations} \label{eq:General ADMM}
    \begin{align}
        \bX^{k+1}
            &=\arg\min_\bX\Big(f(\bX)
                +\frac{\rho}{2}\|\bX
                    -\bZ^k
                    +\bT^k\|_F^2\Big), 
                    \label{eq:General ADMM x Update} \\
        \bZ^{k+1}
            &=\arg\min_\bZ\Big(I_\Omega(\bZ)
                +\frac{\rho}{2}\|\bX^{k+1}
                    -\bZ
                    +\bT^k\|_F^2\Big) \nonumber \\ &= \Pi_\Omega(\bX^{k+1}+\bT^k), 
                    \label{eq:General ADMM z Update} \\
        \bT^{k+1}
            &=\bT^k+(\bX^{k+1}
                -\bZ^{k+1}), 
                \label{eq:General ADMM u Update}
    \end{align}
\end{subequations}
where $\bT\in\bbR^{n\times m}$ is the ``dual" variable, $k\in\bbN$ is the iteration number, $\rho>0$ is a hyper parameter \cite{boyd2011distributed}, and $\Pi_\Omega:\bbR^{n\times m}\rightarrow \Omega$ is the projection operator \cite{boyd2011distributed}. Since the functions $f$ and $I_\Omega$ are closed, proper, and convex, \cref{eq:Equivalent Constrained Optimization Problem} is a convex optimization with equality constraints, which satisfies Slater's condition, implying that its Lagrangian has a saddle point \cite{boyd2004convex}. If \cref{eq:Equivalent Constrained Optimization Problem} is feasible and its Lagrangian has a saddle point, \gls{admm} guarantees that $k\to\infty$ implies $\bX^{k}-\bZ^k\to\bzero$, $f(\bX^k)+I_\Omega(\bZ^k)\to f(\bX^\star)$, and $\bT^k\to\bT^\star$, where $\{\bX^\star,\bT^\star\}$ is a global minimizer of \cref{eq:Constrained Optimization Problem} \cite{boyd2011distributed}.

\subsection{QSR-Dissipativity of Large-Scale Multi-Agent Systems}
QSR-dissipativity, defined below, describes a bound on the relationship between system inputs and outputs.
\begin{definition} [QSR-Dissipativity \cite{vidyasagar1981input}] \label{def:QSR}
    Let $\bQ\in\bbS^{l}$, $\bR\in\bbS^{m}$, $\bS\in\bbR^{l\times m}$. The system $\mathscr{G}:\cL_{2e}^m\to\cL_{2e}^l$ is \textit{QSR-dissipative} if there exists $\beta\in\bbR$ such that for all $\bu\in\cL_2^m$ and $T>0$,
    \begin{align} \label[ineq]{eq:QSR Dissipativity}
        \int_0^T\begin{bmatrix}
            \mathscr{G}^T(\bu(t)) & \bu^T(t)
        \end{bmatrix}
        \begin{bmatrix}
            \bQ & \bS \\ \bS^T & \bR
        \end{bmatrix}
        \begin{bmatrix}
            \mathscr{G}(\bu(t)) \\ \bu(t)
        \end{bmatrix}dt\geq\beta.
    \end{align}
\end{definition}
For \gls{lti} systems, \cref{lem:KYP Lemma} can be used to prove QSR-dissipativity.
\begin{lemma}[\gls{kyp} \cite{gupta1996robust}] \label[lem]{lem:KYP Lemma} 
    An \gls{lti} system with minimal realization $\Sigma{:}\dot{\bx}{=}\bA\bx{+}\bB\bu,\ \by{=}\bC x{+}\bD u$ is QSR-dissipative if there exist matrices $\bP\succ0$, $\bQ$, $\bS$, and $\bR$ satisfying
    \setlength{\arraycolsep}{2.5pt}
    \begin{align} \label[ineq]{eqn:KYP Lemma}
        \begin{bmatrix}
            \bA^T\bP{+}\bP\bA{-}\bC^T\bQ\bC & \bP\bB{-}\bC^T\bS{-}\bC^T\bQ\bD \\
            \bB^T\bP{-}\bS^T\bC{-}\bD^T\bQ\bC & {-}\bR{-}\bS^T\bD{-}\bD^T\bS{-}\bD^T\bQ\bD
        \end{bmatrix}\preceq0.
    \end{align}
\end{lemma}
QSR-dissiaptivity is useful for ensuring $\cL_2$-stability. 
\begin{definition}[$\cL_2$-stability \cite{vidyasagar1981input}] \label{def:L2 Stable}
    An operator $\mathscr{H}:\cL_{2e}\mapsto\cL_{2e}$ is $\cL_2$-stable if there exist $\gamma>0$ and $\beta \in \bbR$ such that
    \begin{align}
        \|(\mathscr{H}\bu)_T\|_2\leq\gamma\|\bu_T\|_2+\beta\quad \forall u\in\cL_2,\ T>0.
    \end{align}
\end{definition}

For multi-agent systems, \gls{vndt}, stated below, relates the dissipativiy of each agent to $\cL_2$ stability of entire system.

\begin{theorem}[\gls{vndt} \cite{vidyasagar1981input}] \label{thm:Interconnected QSR Stability}
    Consider a multi-agent system, $\mathscr{G}:\bu\to\by$, composed of $N$ Q$_i$S$_i$R$_i$-dissipative agents, $\mathscr{G}_i:\cL^{m_i}_{2e}\rightarrow\cL^{l_i}_{2e}$, with mappings $\by_i=\mathscr{G}_i\be_i$ interconnected as
    \begin{align}\label{eq:Interconnected system}
        \begin{split}
            \by=\mathscr{G}\bu, \qquad \be=\bu+\bH\by, 
        \end{split}
    \end{align}
    where $\be=[\be_1^T, \dots, \be_N^T]^T$ is the interconnection signal, $\by = [\by_1^T, \dots, \by_N^T]^T$ is the output, $\bu = [\bu_1^T, \dots, \bu_N^T]^T$ is the exogenous input, and $\bH$ is the interconnection matrix with $(\bH)_{ii}=\bzero$.
    Then $\mathscr{G}$ is $\cL_2$ stable if $\barbQ(\bX)\prec 0$, where
    \begin{align}\label{eq:barQSR}
        \barbQ(\bX)&=\bQ+\textbf{SH}+\bH^T\bS^T+\bH^T\textbf{RH},
    \end{align}
    with $\bX=\diag(\bX_i)_{i\in\bbN_N}$, $\bX_i = \mathrm{diag}(\bP_i,\bQ_i,\bS_i,\bR_i)$, $\bQ=\diag(\bQ_i)_{i\in\bbN_N}$, $\bR=\diag(\bR_i)_{i\in\bbN_N}$, and $\bS=\diag(\bS_i)_{i\in\bbN_N}$.
\end{theorem}

\section{PROBLEM STATEMENT} \label{chap:Problem Statement}
If all agents, $\mathscr{G}_i$, are \gls{lti}, $\cL_2$ the stability criteria of \Cref{thm:Interconnected QSR Stability} can be verified by solving
\begin{subequations} \label[prob]{opt:Main Problem}
    \begin{align}
        \text{Find}\quad&\bX_i \quad i\in\bbN_N,\label{opt:Main Problem: Objective} \\
        \text{s.t.}\quad&\bX_i\in\bbP_i, \label{opt:Main Problem: KYP}  \\
            &\bX\in\tilbbQ, \label{opt:Main Problem: Stability} 
    \end{align}
\end{subequations}
where $\bX$ and $\bX_i$ are defined in \Cref{thm:Interconnected QSR Stability}, $\bbP_i=\{\bX_i\;|\;$ $\mathbf{X}_i \text{ satisfies \cref{lem:KYP Lemma} and }\bP_i\succ0\}$, and  $\tilbbQ=\{\bX\;|\;\barbQ(\bX)\prec0\}$.

Traditionally, \cref{opt:Main Problem} has been solved by finding $\bX_i$ that satisfy \cref{opt:Main Problem: KYP} for each agent, then checking if \cref{opt:Main Problem: Stability} holds. This approach is sub-optimal because all dissipative systems satisfy \cref{eq:QSR Dissipativity} with various $(\bQ,\bS,\bR)$ combinations, and finding compatible dissipativity parameters is complicated with many agents. Therefore, co-optimizing agents' dissipativity in \cref{opt:Main Problem} is more effective than checking $\barbQ(\bX)\prec0$ afterward. Nonetheless, solving \cref{opt:Main Problem} directly introduces several potential drawbacks. First, it requires each agent to share their state space dynamics matrices, $(\bA_i,\bB_i,\bC_i,\bD_i)$, which may be unacceptable due to intellectual property or cyber-security concerns. Second, if the number of agents, $N$, is very large, \cref{opt:Main Problem: Stability} dominates the computational complexity, which becomes $\cO(N^6)$ \cite{nesterov2013introductory}. This paper overcomes these limitations by solving \cref{opt:Main Problem} with distributed optimization methods.

\section{MAIN RESULTS} \label{chap:Main Results}
This section develops two algorithms to solve \cref{opt:Main Problem} in a distributed manner. The first algorithm achieves information security by allowing each agent to calculate its dissipativity parameters independently without sharing its dynamics. These dissipativity parameters are shared and iteratively adjusted to satisfy \gls{vndt}. Because there are uncountably many different systems with the same dissipativity parameters, this maintains much greater information security than sharing dynamics information. However, this first approach involves a problem with a \gls{lmi} constraint whose dimension scales with the number of agents.
To remedy this, the second algorithm applies \Cref{thm:Chordal Decomposition} to \gls{vndt}. This algorithm reduces computational complexity by reformulating the largest \gls{lmi} into smaller \glspl{lmi} and a single equation. It also limits the sharing of certain dissipativity parameters, which further enhances information security.

\subsection{Information-Secure Network Stability Analysis} \label{chap:Distributed Stability Analysis}
Using indicator functions, \cref{opt:Main Problem} is equivalent to 
\begin{align} \label[prob]{opt:Distributed Optimization}
    \arg\min_{\bX,\bZ}\quad\mathlarger{\Sigma}_{i\in\bbN_N}I_{\bbP_i}(\bX_i)+I_{\tilbbQ}(\bZ)\quad\text{s.t.}\quad\bX-\bZ=\textbf{0},
\end{align}
with the primal and clone variable, $\bX$ and $\bZ=\diag(\bZ_i)_{i\in\bbN_N}$. This instantiates \cref{eq:Equivalent Constrained Optimization Problem} with $f(\bX)=0$, which is convex, 
%
so \cref{opt:Distributed Optimization} can be solved using \gls{admm} by iterating as
\begin{subequations} \label{eq:Distributed ADMM}
    \begin{align}
        \bX_i^{k+1}&=\Pi_{\bbP_i}(\bZ_i^k-\bT_i^k), \quad
            i\in\bbN_N,\label{eq:X proj}\\
        \bZ^{k+1}&=\Pi_{\tilbbQ}(\bX^{k+1}+\bT^k), \label{eq:Z proj} \\ 
        \bT_i^{k+1}&=\bT_i^k+(\bX_i^{k+1}-\bZ_i^{k+1}), \quad
            i\in\bbN_N, \label{eq:T min}
    \end{align}
\end{subequations}
where $\bT=\diag(\bT_i)_{i\in\bbN_N}$ acts as the dual variable. \cref{alg:01} describes the iterative optimization process of using \cref{eq:Distributed ADMM}. The initial points satisfy $\bX^0=\bZ^0$ and $\bT^0=\textbf{0}$. Any choice is acceptable for the initial $\bX_i^0$. A natural choice is $\mathbf{X}^0_i\in\bbP_i$. Since \cref{alg:01} seeks a feasible point, rather than an optimal one, $\barbQ(\bX)\prec0$ acts as the stopping criterion. If a feasible point exists for \cref{opt:Main Problem}, then all assumptions for the convergence of \gls{admm} hold, so \cref{alg:01} converges to a feasible point, confirming network stability. If \cref{alg:01} converges, \gls{vndt} ensures stability. Otherwise, \gls{vndt} may not guarantee the stability of the network.

Sometimes, only a single $\bQ\bS\bR$ property can be established for an agent. In such cases, the computation of \cref{eq:X proj} can be bypassed by directly assigning $\bX_i$ to the corresponding dissipaitivy parameters.  For example if the $i^\thh$ agent is an RLC circuit of unknown parameters, it is passive, and that is all that is known. Then \cref{eq:X proj} can be substituted with $\bX_i=\diag(\bP_i,\textbf{0},\bI,\textbf{0})$, for any $\bP_i\succ0$.

\begin{algorithm}
    \caption{Information-secure network stability analysis}\label{alg:01}
    \begin{algorithmic}[1]
        \Require MaxIter$,\bA_i,\bB_i,\bC_i,\bD_i,\bX_i^0$ for $i\in\bbN_N$
        \Ensure $\bX^k$
        \State Construct dynamics using $\bA_i,\bB_i,\bC_i,\bD_i$ for $i\in\bbN_N$
        \State Initialize $k=0$, $\bX^0=\diag(\bX_i^0)_{i\in\bbN_N}$, $\bZ^0=\bX^0$, and $\bT^0=\textbf{0}$
        \While {$\barbQ(\bX)\nprec0,k<$MaxIter}
            \State $k\gets k+1$
            \State Find $\bX_i^k$ by \cref{eq:X proj} in parallel
            \State Find $\bZ^k$ by \cref{eq:Z proj} at a centralized node
            \State Find $\bT_i^k$ by \cref{eq:T min} in parallel
        \EndWhile
        \If {$k<$MaxIter}
            \State Multi-agent system is stable with $\bX^k=\diag(\bX_i^0)_{i\in\bbN_N}$
        \EndIf
    \end{algorithmic}
\end{algorithm}

\subsection{Chordal Decomposition of \texorpdfstring{\gls{vndt}}{VNDT}} \label{SubChap: Chordal Decomposition}

The key advantage of \cref{alg:01} is that each agent independently verifies its own dissipativity without sharing its dynamics. The calculation time of $k^\thh$ iteration is $t_{k}=\max_{i\in\bbN_N}(t_k^{X_i})+t_k^Z+\max_{i\in\bbN_N}(t_k^{T_i})$, where $t_k^Z$ is the time required for \cref{eq:Z proj} and $t_k^{X_i}$ and $t_k^{T_i}$ are the times required respectively for \cref{eq:X proj} and \cref{eq:T min} at the $i^{th}$ agent. In multi-agent systems, $t_k^Z$ dominates the overall computational cost, as the size of \cref{eq:Z proj} scales with the size of the network, $N$, and its complexity when using interior point methods scales as $\cO(N^6)$ \cite{nesterov2013introductory}. The other sub-problems scale with the size of each subsystem, $n_i$, which is much smaller than $N$. Therefore, the bottleneck is solving \cref{eq:Z proj}, which can be eliminated by decomposing $\barbQ$.

The terms in \cref{tab:Chordal Decomposition} are used to explain the decomposition procedure. From \cref{eq:barQSR}, each component of $\barbQ(\bX)$ is 
\begin{subnumcases}{\label{eq:barQ} (\barbQ)_{ij}{=}} 
        \begin{aligned}
        &\bQ_i{+}\mathlarger{\Sigma}_{k\in\cV_i(\cH)}(\bH)_{ki}^T\bR_k(\bH)_{ki}, 
        \end{aligned}\quad\quad i=j, \\
        \begin{aligned}
        &\bS_i(\bH)_{ij}
            {+}(\bH)_{ji}^T\bS_j^T \\
            &\quad{+}\mathlarger{\Sigma}_{k{\in}\cV_i(\cH){\cap}\cV_j(\cH)}(\bH)_{ki}^T\bR_k(\bH)_{kj}, 
        \end{aligned}\; i\neq j,
\end{subnumcases}
where $\bQ_i,\bS_i$, and $\bR_i$ are dissipativity parameters of $i^\thh$ agent, $\bH$ is the interconnection matrix of the network, and $\cV_i(\cH)=\set{j\in\cV(\cH)\,|\,(i,j)\in\cE(\cH)}$. Since $\barbQ(\bX)$ is symmetric, its structure graph, $\barcQ$, is undirected and $\cH\subset\barcQ$ and $\cE(\barcQ)=\cE(\cH)+\set{(i,j)\,|\,i,j\in\cV(\cH),\cV_i(\cH)\cap\cV_j(\cH)\neq\varnothing}$.

\begin{table}
\caption{Terms for decomposition of \cref{eq:barQSR}}
\label{tab:Chordal Decomposition}
\begin{center}
\renewcommand{\arraystretch}{1.2}
\begin{tabular}{p{0.75cm}|p{6.75cm}}
\hline
Terms & Definition \\
\hline\hline
$M$ & The number of maximal cliques of $\barcQ$ \\
\hline
$C_p(\barcQ)$ & $p^\text{th}$ maximal clique of $\barcQ$   \\
\hline
$\cC_p$ & The complete graph from $C_p(\barcQ)$, which has $\cV(\cC_p)=C_p(\barcQ)$ and $\cE(\cC_p)=\set{(i,j)\in\cE(\barcQ)\,|\,i,j\in C_p(\barcQ)}$ \\ \hline
$\barcQ_o$ & 
$\bigcup_{p,q\in\bbN_M}(\cC_p\cap\cC_q)$\\
\hline
$L$ & $|\cE(\barcQ_o)|$; the number of overlapped edges from \Cref{thm:Chordal Decomposition}. \\
\hline
\vfill
$\cV(\mathcal{R})$ & 
$\set{k{\in}\cV_i(\cH){\cap}\cV_j(\cH){|}{\forall}(i,j){\in}\cE(\barcQ_o)}$; the set of agents whose $\bR$ matrix is used to calculate $(\barbQ)_{ij}$ for $(i,j)\in\cE(\barcQ_o)$
\\ \hline
\vfill
$\bY_{ij}^{p}$ & The new matrix variables defined from the overlapped position in \cref{eq:Chordal Decomposition}, where $p\in\bbN_M$ is the index of maximal cliques and $(i,j)\in\cE(\barcQ_o)$. \\
\hline
\vfill
$\bY$ & $\diag\big(\diag(\bY_{ij}^p)_{p{\in}\{p{\in}\bbN_M{|}(i,j){\in}\cE(\barcQ_o{\cap}\cC_p)\}}\big)_{(i,j){\in}\cE(\barcQ_o)}$;~the~block diagonal matrix defined from all overlapped variables $\bY_{ij}^p$. \\
\hline
\vfill
$\barbY_p$ & $\diag(\bY_{ij}^p)_{(i,j)\in\cE(\barcQ_o\cap\cC_p)}$; the block diagonal matrix with block diagonal components in $\bY$, using vertices in $C_p(\barcQ)$. \\
\hline
\vfill
$\barbX_p$ & The block diagonal matrix defined from block diagonal components in $\bX$ which are used to calculate $(\barbQ)_{ij}$ in \cref{eq:barQ}, where $(i,j)\in\cE(\cC_p)-\cE(\barcQ_o)$. \\
\hline
\vfill
$\hatbX$ & $\diag\big(\diag(\bQ_i)_{i\in\cV(\barcQ_o)},\diag(\bS_i)_{i\in\cV(\barcQ_o)},\diag(\bR_i)_{i\in\cV(\mathcal{R}) }\big)$; the block diagonal matrix defined from dissipativity matrices used to calculate $(\barbQ)_{ij}$ for $(i,j)\in\cE(\barcQ_o)$. \\
\hline
\end{tabular}
\end{center}
\vspace*{-1.25\baselineskip} 
\end{table}

In practice, the constraint $\barbQ(\bX)\prec0$ can be replaced by $\barbQ(\bX)+\epsilon\bI\preceq0$ with $\epsilon>0$.
From \Cref{thm:Chordal Decomposition}, we have  $\barbQ(\bX)+\epsilon\bI=\sum_{p=1}^M\bE_{C_p(\barcQ)}^T\barbQ_p\bE_{C_p(\barcQ)}$. 
By expressing $\barbQ(\bX)+\epsilon\bI$ as the sum of smaller negative semi-definite matrices based on the maximal cliques of $\barcQ$, as exemplified in \cref{fig:Chordal Decomposition}, $\barbQ(\bX)+\epsilon\bI\preceq0$ can be decomposed into the $M$ smaller \glspl{lmi} and $L$ equality constraints,
\begin{align}
    \barbQ_p(\barbX_p,\barbY_p)&\preceq0,\quad\;\; \forall p\in\bbN_M,  \label[ineq]{eq:Chordal LMIs}\\
    \mathlarger{\Sigma}_{\set{p\in\bbN_M|(i,j)\in\cE(\barcQ_o\cap\cC_p)}} \bY_{ij}^{p} &{=} (\barbQ)_{ij},\;\forall(i,j)\in\cE(\overline{Q}_o).  \label{eq:Chordal Equality}
\end{align}

\begin{figure}
    \centering
    \includegraphics[width = 0.4\textwidth]{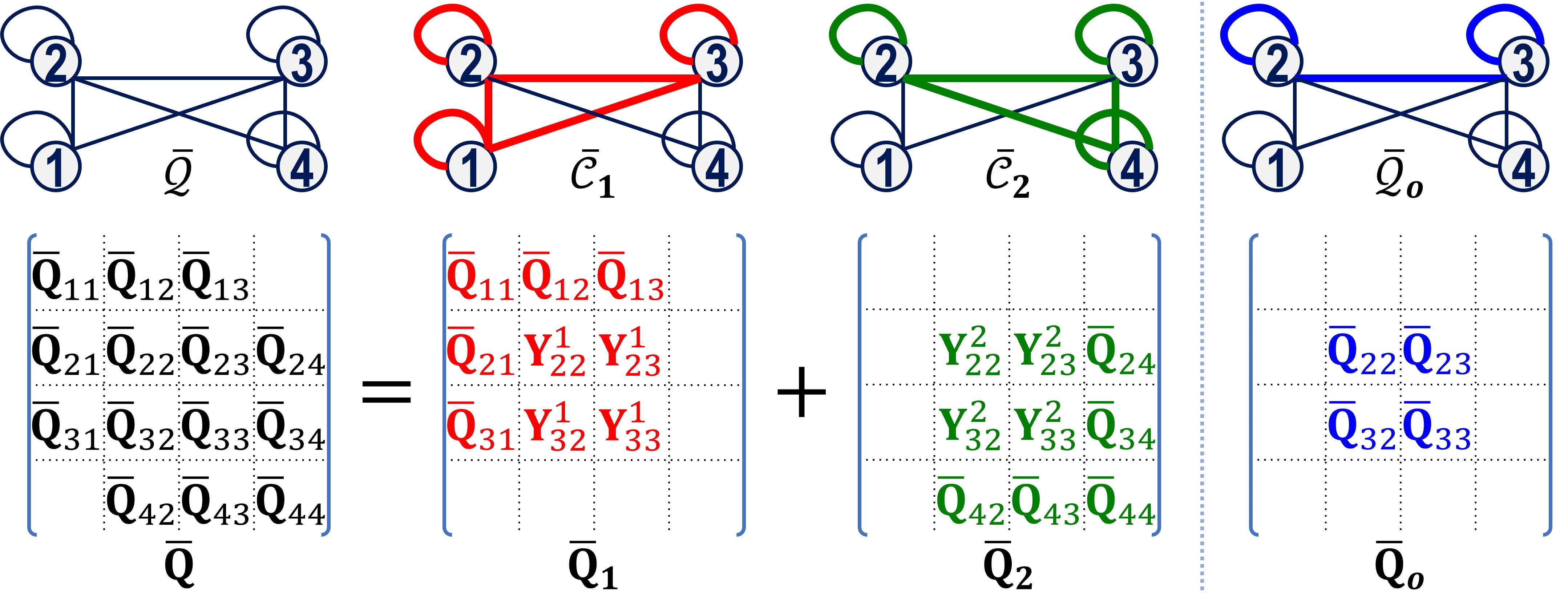}
    \caption{Example of chordal decomposition of graph $\barcQ$: $\barcQ_o$ represents the overlapped graph from \Cref{thm:Chordal Decomposition}.}
    \label{fig:Chordal Decomposition}
    \vspace*{-1.25\baselineskip} 
\end{figure}

Each of the \glspl{lmi} in \cref{eq:Chordal LMIs} are independent, meaning that $\barbX_p$ and $\barbX_q$ for $p\neq q$ contain different $(\bQ_i,\bS_i,\bR_i)$ parameters. The following theorem shows that $\big(\cV_i(\cH)\cap\cV_j(\cH)\big)\cap\big(\cV_k(\cH)\cap\cV_l(\cH)\big)=\varnothing$ for $i,j\in C_p(\barcQ)$ and $k,l\in C_q(\barcQ)$, which implies this independence.

\begin{theorem} \label{thm:No Overlap}
    Consider $\barbQ(\bX)$ from \cref{eq:barQ}. Using the notation of \cref{tab:Chordal Decomposition}, for all $i,j\in C_p(\barcQ)$ and $k,l\in C_q(\barcQ)$ with $p\neq q$, if $(i,j),(k,l)\notin\cE(\barcQ_o)$, then $\big(\cV_i(\cH)\cap\cV_j(\cH)\big)\cap\big(\cV_k(\cH)\cap\cV_l(\cH)\big)=\varnothing$.
\end{theorem}

\begin{proof}
    \Cref{thm:No Overlap} can be proved using its contrapositive. Assume that $\big(\cV_i(\cH)\cap\cV_j(\cH)\big)\cap\big(\cV_k(\cH)\cap\cV_l(\cH)\big)\neq\varnothing$; that is, the intersections $\cV_i(\cH)\cap\cV_k(\cH)$, $\cV_i(\cH)\cap\cV_l(\cH)$, $\cV_j(\cH)\cap\cV_k(\cH)$, and $\cV_j(\cH)\cap\cV_l(\cH)$ are non-empty. This implies that $(\barbQ)_{ik}\neq\textbf{0}$, $(\barbQ)_{il}\neq\textbf{0}$, $(\barbQ)_{jk}\neq\textbf{0}$, and $(\barbQ)_{jl}\neq\textbf{0}$, meaning that $i,j\in C_q(\barcQ)$ and $k,l\in C_p(\barcQ)$. Since $i,j\in C_p(\barcQ)$ and $k,l\in C_q(\barcQ)$, it follows that $(i,j),(k,l)\in\cE(\barcQ_o)$.
\end{proof}

Due to \Cref{thm:No Overlap}, \cref{opt:Main Problem: Stability}, $\barbQ(\bX)+\epsilon\bI{\preceq}0$, can be divided into smaller inequality constraints, $\diag(\barbX_p,\barbY_p)\in\barbbQ_p$ for all $p\in\bbN_M$, and single equality constraint, $\diag(\hatbX,\bY)\in\hatbbQ$, where $\barbbQ_p=\set{\diag(\barbX_p,\barbY_p)\;|\;\barbQ_p(\barbX_p,\barbY_p)\preceq0}$ and $\hatbbQ=\set{\diag(\hatbX,\bY)\;|\;\sum_{\set{p\in\bbN_M|(i,j)\in\cE(\barcQ_o\cap\cC_p)}} \bY_{ij}^{p} {=} (\barbQ)_{ij}\,\forall(i,j)\in\cE(\barcQ_o)}$. Then, \cref{opt:Main Problem} is equivalently formulated as
\begin{subequations} \label[prob]{opt:Main Chordal Problem}
    \begin{align}
        \text{Find}\quad&\bX_i \quad i\in\bbN_N,\label{opt:Main Chordal Problem: Objective} \\
        \text{s.t.}\quad&\bX_i\in\bbP_i, \label{opt:Main Chordal Problem: KYP} \\ 
            &\diag(\barbX_p,\barbY_p)\in\barbbQ_p,\quad p\in\bbN_M, \label{opt:Main Chordal Problem: Stability}  \\
            &\diag(\hatbX,\bY)\in\hatbbQ, \label{opt:Main Chordal Problem: Equality}
    \end{align}
\end{subequations}
where $\bX_i$, $\bbP_i$ are defined in \cref{opt:Main Problem}.

\begin{table}
\caption{Variables for \gls{admm} of \cref{opt:Main Chordal Problem}}
\label{tab:ADMM Variables}
\begin{center}
\renewcommand{\arraystretch}{1.2}
\begin{tabular}{p{1cm}|p{6.5cm}}
\hline
 Variables & Definition \\
\hline\hline
$\bZ_p$ & The clone variable of $\diag(\barbX_p,\barbY_p)$\\
\hline
$\bW$ & The clone variable of $\diag(\hatbX,\bY)$\\
\hline
$\bJ$ & The global clone variable of $\diag(\bX,\bY)$\\
\hline
$\bT_i$ & The dual variable of $\bX_i$ \\
\hline
$\bU_p$  & A dual variable of $\bZ_p$ \\
\hline
$\bV$    & A dual variable of $\bW$ \\
\hline
\vfill
$\tilbJ_i$  & The block diagonal matrix defined from block diagonal components in $\bJ$, which serve as a clone variable of $\bX_i$ \\
\hline
\vfill
$\barbJ_p$   & The block diagonal matrix defined from block diagonal components in $\bJ$, which serve as a clone variable of $\bZ_p$ \\
\hline
\vfill
$\hatbJ$      & The block diagonal matrix defined from block diagonal components in $\bJ$, which serve as a clone variable of $\bW$ \\
\hline
\end{tabular}
\end{center}
\vspace*{-1.25\baselineskip} 
\end{table}

\subsection{Distributed Network Stability Analysis} \label{SubChap: Chordal}
By introducing clone and dual variables in \cref{tab:ADMM Variables}, \Cref{opt:Main Chordal Problem} is equivalent to 
\begin{subequations} \label[prob]{opt:Eqauivalent Main Chordal Problem}
    \begin{align}
        \arg\min_{\bX,\bZ,\bW,\bJ}\quad&
            \mathlarger{\Sigma}_{i{\in}\bbN_N}I_{\bbP_i}(\bX_i)
            {+}\mathlarger{\Sigma}_{p{\in}\bbN_M}I_{\barbbQ_p}(\bZ_p)
            {+}I_{\hatbbQ}(\bW), \label{opt:Equivalent Main Problem: Objective} \\
        \text{s.t.}\quad
            &\bX_i-\tilbJ_i=\textbf{0},\quad i\in\bbN_N, \label{opt:Equivalent Main Chordal Problem: KYP} \\ 
            &\bZ_p-\barbJ_p=\textbf{0},\quad p\in\bbN_M, \label{opt:Equivalent Main Chordal Problem: Stability}  \\
            &\bW-\hatbJ=\textbf{0}. \label{opt:Equivalent Main Chordal Problem: Equality}
    \end{align}
\end{subequations}

Now, \gls{admm} can be applied to \cref{opt:Eqauivalent Main Chordal Problem}.
The \gls{admm} solution to \cref{opt:Eqauivalent Main Chordal Problem} can be stated in three steps as follows, where $k\in\bbN$ is the iteration number.

\subsubsection{Primal \textbf{X}, Clones \textbf{Z}, \textbf{W} Update}
Holding $\bJ^k$, $\bT^k$, $\bU^k$, and $\bV^k$, constant, $\bX^{k+1}$, $\bZ^{k+1}$ and $\bW^{k+1}$ are updated as
\begin{subequations} \label{eq:Chordal proj}
    \begin{align}
        \bX_i^{k+1}
            &=\Pi_{\bbP_i}(\tilbJ_i^k-\bT_i^k), 
            \quad i\in\bbN_N, \label{eq:X Chordal proj}\\
        \bZ_p^{k+1}
            &=\Pi_{\barbbQ_p}(\barbJ_p^{k}-\bU_p^k) 
            \quad p\in\bbN_M, \label{eq:Z Chordal proj} \\ 
        \bW^{k+1}
            &=\Pi_{\hatbbQ}(\hatbJ^k-\bV^k). \label{eq:W Chordal proj} 
    \end{align}
\end{subequations}

\begin{figure}
    \centering
    \includegraphics[width = 0.475\textwidth]{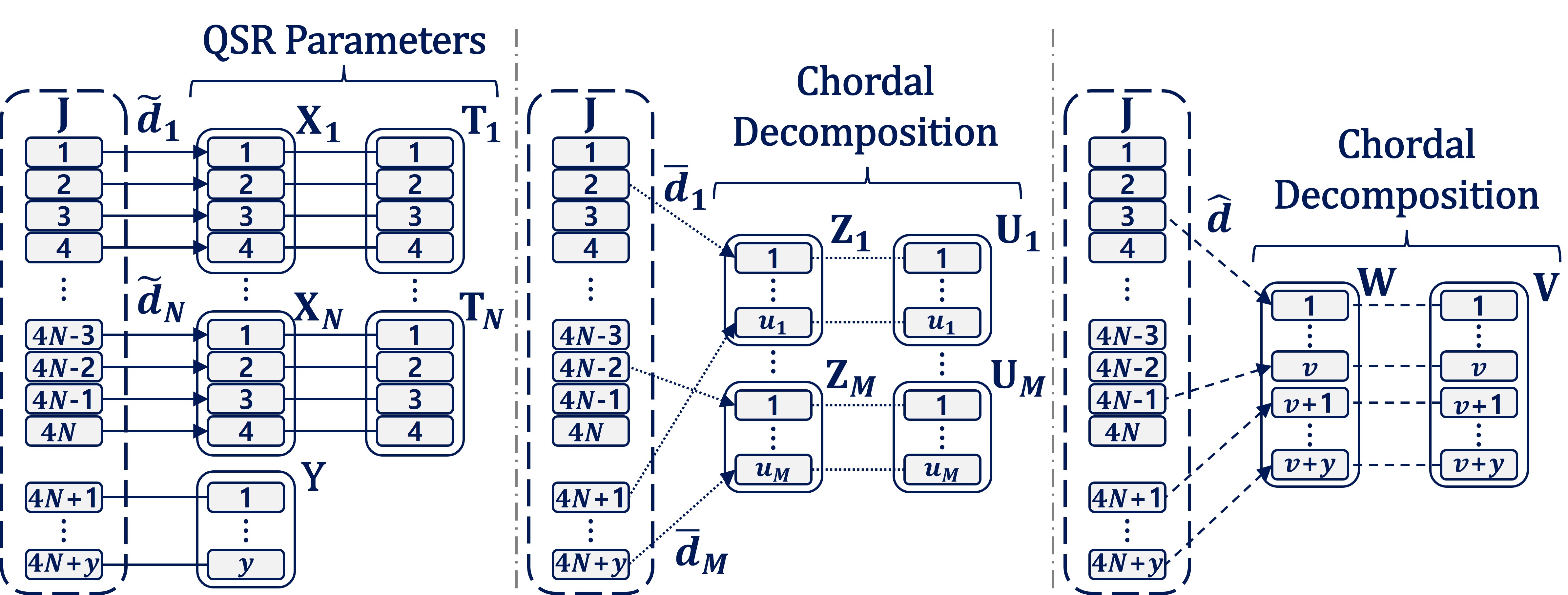}
    \caption{Mapping among all block diagonal components of variables in \cref{tab:ADMM Variables}: The gray rectangles represent the block component indices defining each variable. The QSR parameters are grouped by agent, with blocks for each dissipativity matrix, while the chordal parameters are grouped by clique with blocks for each \gls{lmi} and a matrix equation}.
    \label{fig:Matching}
    \vspace*{-1.25\baselineskip} 
\end{figure}

\subsubsection{Clone \textbf{J} Update}
Holding $\bX^{k+1}$, $\bZ^{k+1}$, $\bW^{k+1}$, $\bT^k$, $\bU^k$, and $\bV^k$, the clone variable $\mathbf{J}^{k+1}$ is updated 
according to
\begin{align} \label{eq:J Chordal bulk}
    \bJ^{k+1}=\arg\min_\bJ
    \left(\begin{array}{l}
         \sum_{i\in\bbN_N}
                        \|\bX_i^{k+1}-\tilbJ_i+\bT_i^k\|_F^2 \\
            \,+\sum_{p\in\bbN_M}
                        \|\bZ_p^{k+1}-\barbJ_p+\bU_p^k\|_F^2 \\
            \,\,+\|\bW^{k+1}-\hatbJ+\bV_r^k\|_F^2
    \end{array}\right).
\end{align}
\cref{eq:J Chordal bulk} is similar to \cref{eq:General ADMM z Update} with $I_\Omega(\bJ)=\textbf{0}$.

To solve \cref{eq:J Chordal bulk}, we define mappings from block indices in the global clone to agent- and clique-wise block diagonal matrix indices. To do this, first note that diagonal blocks of $\widetilde{\bJ}_i$ correspond to matrices in \cref{eqn:KYP Lemma} as $\widetilde{\bJ}_i=\diag(\bP_i,\bQ_i,\bS_i,\bR_i)$, so $\widetilde{d}_a=\{ (i,j)\in \bbN_N\times\bbN_4 | \bJ_a=(\widetilde{\bJ}_i)_{jj} \}$. Likewise, the block diagonals of $\overline{\bJ}_p$ correspond to different \glspl{lmi} in \cref{eq:Chordal LMIs} within the $p^{th}$ clique, so $\overline{d}_a=\{ (p,j)\in \bbN_M\times\bbN_{u_p} | \bJ_a=(\overline{\bJ}_p)_{jj} \}$, where $u_p$ represents the number of block matrices in $\barbX_p$. Lastly, $\widehat{d}_a=\{ (j)\in \bbN_{\nu+y} | \bJ_a=(\hat{\bJ})_{jj} \}$, corresponding to \cref{eq:Chordal Equality}, where $v$ and $y$ denote the number of block matrices associated with $\hatbX$ and $\bY$, respectively. With this, the closed-form solution of \cref{eq:J Chordal bulk} for each block of $\bJ$, is
\begin{align} \label{eq:J Chordal update}
    \bJ_a^{k+1}&{=}
    \frac{
        \displaystyle\sum_{(i,j)\in\widetilde{d}_a}\hspace{-8pt}
                    {(\bX_i^{k+1}{+}\bT_i^k)_{jj}}
        {+}\hspace{-8pt}\sum_{(p,j)\in\overline{d}_a}\hspace{-10pt}
                    {(\bZ_p^{k+1}{+}\bU_p^k)_{jj}}
        {+}\hspace{-3pt}\sum_{j\in\hat{d}_a}\hspace{-3pt}
                (\bW^{k+1}{+}\bV^k)_{jj}
                }{|\widetilde{d}_a|+|\overline{d}_a|+|\hat{d}_a|}.
\end{align}

\cref{eq:J Chordal update} computes the 
average of the corresponding variables. This process leads all variables, independently computed from \cref{eq:X Chordal proj,eq:Z Chordal proj,eq:W Chordal proj,eq:T Chordal update,eq:U Chordal update,eq:V Chordal update}, to a network-wide solution satisfying \cref{opt:Main Chordal Problem: KYP,opt:Main Chordal Problem: Objective,opt:Main Chordal Problem: Stability}.

\subsubsection{Dual Update}
Holding $\bX^{k+1}$, $\mathbf{J}^{k+1}$, $\bZ^{k+1}$, and $\bW^{k+1}$, constant, $\bT^{k+1}$, $\bU^{k+1}$, and $\bV^{k+1}$, are updated according to 
\begin{subequations}\label{eq:Dual Chordal update}
    \begin{align}
        \bT_i^{k+1}&=\bT_i^k+(\bX_i^{k+1}-\tilbJ_i^{k+1}) \label{eq:T Chordal update} \\
        \bU_p^{k+1}&=\bU_p^k+(\bZ_j^{k+1}-\barbJ_p^{k+1}) \label{eq:U Chordal update} \\
        \bV^{k+1}&=\bV^k+(\bW^{k+1}-\hatbJ^{k+1}), \label{eq:V Chordal update}
    \end{align}
\end{subequations}
to ensure convergence of the iterations. The implementation of these \gls{admm} steps is summarized in \cref{alg:02}.

\begin{algorithm}
    \caption{Distributed network stability analysis}\label{alg:02}
    \begin{algorithmic}[1]
        \Require MaxIter$,\bA_i,\bB_i,\bC_i,\bD_i,\bX_i^0,\epsilon$ for $i\in\bbN_N$
        \Ensure $\bX^k$
        \State Construct dynamics using $\bA_i,\bB_i,\bC_i,\bD_i$ for $i\in\bbN_N$
        \State Initialize $k=0$, $\bJ^0=\diag(\diag(\bX_i^0)_{i\in\bbN_N},\bI^\bY)$, where $\bX^0=\tilbJ^0$, $\bZ^0=\barbJ^0$, and $\bW^0=\hatbJ$, 
            $\bT^0=\bU^0=\bV^0=\textbf{0}$
        \While {$\barbQ_p(\barbX_p,\barbY_p)\nprec0$, $\hatbQ(\hatbX,\bY)\neq\textbf{0},k<$MaxIter}
            \State $k\gets k+1$
            \State Find $\bX_i^k,\bZ_p^k,\bW^k$ by \cref{eq:X Chordal proj,eq:Z Chordal proj,eq:W Chordal proj} in parallel
            \State Find $\bJ_\alpha^k$ by \cref{eq:J Chordal update} in parallel
            \State Find $\bT_i^k,\bU_p^k,\bV^k$ by \cref{eq:T Chordal update,eq:U Chordal update,eq:V Chordal update} in parallel
        \EndWhile
        \If {$k<$MaxIter}
            \State Multi-agent system is stable with $\bX^k=\diag(\bX_i^k)_{i\in\bbN_N}$
        \EndIf
    \end{algorithmic}
\end{algorithm}

Like \cref{alg:01}, any initial point where $\bX^0=\tilbJ^0$, $\bZ^0=\barbJ$, $\bW^0=\hatbJ^0$, and $\bT^0=\bU^0=\bV^0=\bzero$ can be used to initialize the algorithm, and a natural choice is $\bJ^0=\diag(\diag(\bX_i^0)_{i\in\bbN_N},\bI^\bY)$, where $\mathbf{X}^0_i\in\bbP_i$ and $\bI^\bY$ is an identity matrix of the same dimension as $\bY$ in \cref{tab:Chordal Decomposition}. \Cref{alg:02} converges to a feasible point of \cref{opt:Main Chordal Problem} if a feasible point exists due to \cite[Appx.]{boyd2011distributed}. Conversely, if the algorithm does not converge, then a feasible point does not exist, meaning the \gls{vndt} does not guarantee the stability of the network.

The computation time for each iteration is given by $t_k=t_k^P+t_k^J+t_k^D$, where $t_k^P=\max_{i\in\bbN_N,j\in\bbN_M}(t_k^{\bX_i},t_k^{\bZ_j},t_k^{\bW})$, $t_k^J=\max_{a\in\bbN_{4N+y}}(t_k^{\bJ_a})$, $t_k^D=\max_{i\in\bbN_N,j\in\bbN_M}(t_k^{\bT_i},t_k^{\bU_j},t_k^{\bV})$. Since all of the projection processes can be executed in parallel, the calculation time is dramatically reduced compared to \cref{alg:01}, although more iterations may be required for convergence. This advantage will be further amplified in larger systems. The numerical example provided in \cref{chap:Numerical Example} will demonstrate the calculation time reduction using this algorithm.

\subsection{Extension to Nonlinear Systems}
If $\mathscr{G}_i$ is not \gls{lti} but has some structured nonlinearity, a variation of \cref{lem:KYP Lemma} may be applied with adjusted definitions of $\bX_i$ and $\bbP_i$. For instance, if $\mathscr{G}_i$ is \gls{lti} with polytopic uncertainty, then $\bX_i$ remains the same, and $\bbP_i = \{\bX_i \,|\, \bX_i \text{ satisfies \cite[Equation 18]{walsh2019interior}}\}$. Alternatively, if $\mathscr{G}_i$ is \gls{lti} with a constant state delay, then $\bX_i$ is augmented with some additional variables, and $\bbP_i = \{\bX_i \,|\, \bX_i \text{ satisfies \cite[Equation 5]{li2002delay}}\}$. Critically, $\tilbbQ$ only depends on $\bQ_i$, $\bS_i$, and $\bR_i$, so these variations in $\bX_i$ and $\bbP_i$ do not affect \cref{opt:Main Problem: Stability,opt:Main Chordal Problem: Stability,opt:Main Chordal Problem: Equality}. Therefore, the results extend to these and other structured nonlinear cases.

\section{NUMERICAL EXAMPLE} \label{chap:Numerical Example}
\begin{figure}
    \centering
    \includegraphics[width = 0.45\textwidth]{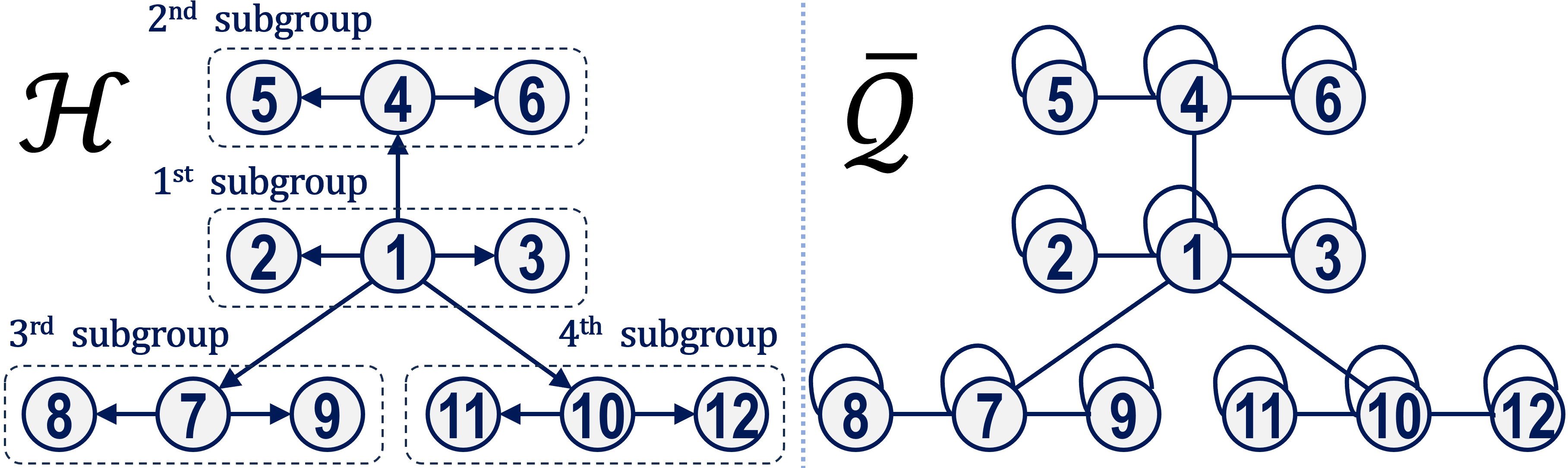}
    \caption{Graph network of \gls{uav}.}
    \label{fig:Graph}
    \vspace*{-1.25\baselineskip} 
\end{figure}

The stability of a 2D swarm of \glspl{uav} is analyzed using \cref{alg:01,alg:02} to demonstrate the proposed approach. The system consists of 4 subgroups, each containing 3 \glspl{uav}. \Cref{fig:Graph} illustrates the network $\cH$ and $\barcQ$. The graph $\barcQ$ features 11 maximal cliques, resulting in 11 \glspl{lmi}. The dynamics of each subgroup, $i\in\bbN_4$, is given by
\setlength{\arraycolsep}{0pt}
\begin{align*}
    \dot{\bx}_i
    &{=}{\begin{bmatrix}
        \bA_{cl} & 0 & 0 \\
        \bB\bK & \bA_{cl} & 0 \\
        \bB\bK & 0 & \bA_{cl}
    \end{bmatrix}}
    \bx_i {+}(\bB\bK\otimes\bI_3)\bx_i^d
    {+}(1-\delta_1(i))\begin{bmatrix}
        \bB\bK & 0 & 0 \\
        0 & 0 & 0 \\
        0 & 0 & 0
    \end{bmatrix}\bx_1
\end{align*}
where $\bA_{cl}=\bA{-}\bB\bK$, $\delta_1(1)=1$, $\delta(i)=0$ if $i\neq1$,
\setlength{\arraycolsep}{2pt}
\begin{align*}
    \bA&= \begin{bmatrix}  0_{3\times 2} & 0_{3\times 1} & I_3 \\ 0_{1\times 2} & -g & 0_{1\times 3} \\ 0_{2\times 2} & 0_{2\times 1} & 0_{2\times 3}\end{bmatrix},
    \bB=\begin{bmatrix}
        0_{4\times 1} & 0_{4\times 1} \\
        \frac{1}{m} & \frac{1}{m} \\
        -\frac{l}{I_{xx}} & \frac{l}{I_{xx}}
    \end{bmatrix},\\
    \bK&=\begin{bmatrix*}[r]
        7.07 & 7.07 & -49.00 & 8.70 & 5.12 & -15.81 \\
        -7.07 & 7.07 & 49.00 & -8.70 & 5.12 & 15.81 
    \end{bmatrix*},
\end{align*}
$g$ is gravitational acceleration, $\otimes$ is Kronecker product, $\bx_i\in\bbR^{18}$ is the states, $\bx_i^d$ is the desired position input. The system and input matrices $\bA$ and $\bB$ are derived from 2D \gls{uav} dynamics, with mass $m=3$ kg, a moment of inertia $I_{xx}=1$ kg-m$^2$, and wing length $l=0.2$ m. The gain matrix $\bK$ is computed using a linear-quadratic regulator controller.

\cref{alg:01,alg:02} were applied with an initial point $\bJ^0=\diag(\diag(\bX_i^0)_{i\in\bbN_{12}},\bI_6\otimes\bI_{14})$ for $\mathbf{X}^0_i\in\bbP_i$. 
The initial points for $\bX$, $\bZ$, $\bW$, $\bT$, $\bU$, and $\bV$ followed the initialization rules in \cref{chap:Main Results}. Both algorithms found the feasible $\bX_i$ with $\barbQ(\bX)\prec0$ without sharing the dynamics information.

\begin{figure}
\centering
    \begin{subfigure}[t]{0.23\textwidth}
        \centering
        \resizebox{\textwidth}{!}{
%
%
\definecolor{mycolor1}{rgb}{0.00000,0.44700,0.74100}%
\definecolor{mycolor2}{rgb}{0.85000,0.32500,0.09800}%

\begin{tikzpicture}

\begin{axis}[%
width=5in,
height=2.5in,
at={(0.758in,0.509in)},
scale only axis,
bar shift auto,
xmin=1.5,
xmax=13.5,
xtick={3, 6, 9, 12},
xlabel style={font=\color{white!15!black}},
xlabel={Number of UAVs},
ymin=0,
ymax=3000,
ylabel style={at={(axis description cs:-0.125,.5)},font=\color{white!15!black}},
ylabel={Iterations},
axis background/.style={fill=white},
title style={font=\bfseries},
title={\textbf{Iterations Required}},
xmajorgrids,
ymajorgrids,
legend style={at={(0.03,0.97)}, anchor=north west, legend cell align=left, align=left, draw=white!15!black},
title style={font=\Huge},xlabel style={font={\Huge}},ylabel style={font=\Huge},legend style={font=\Huge},tick label style={font=\LARGE}
]
\addplot[ybar, bar width=1, fill=green, draw=black, area legend] table[row sep=crcr] {%
3	218\\
6	752\\
9	1002\\
12	1226\\
};
\addlegendentry{alg 1}

\addplot[ybar, bar width=1, fill=yellow, postaction={
        pattern=north east lines
    }, draw=black, area legend] table[row sep=crcr] {%
3	430\\
6	1746\\
9	2136\\
12	2760\\
};
\addlegendentry{alg 2}

\end{axis}
\end{tikzpicture}%
}
    \end{subfigure}
    \begin{subfigure}[t]{0.245\textwidth}
        \centering
        \resizebox{\textwidth}{!}{
%
%
\definecolor{mycolor1}{rgb}{0.00000,0.44700,0.74100}%
\definecolor{mycolor2}{rgb}{0.85000,0.32500,0.09800}%
\begin{tikzpicture}

\begin{axis}[%
width=5in,
height=2.5in,
at={(0.758in,0.509in)},
scale only axis,
bar shift auto,
xmin=1.5,
xmax=13.5,
xtick={3,6,9,12},
xlabel style={font=\color{white!15!black}},
xlabel={Number of UAVs},
ymin=0,
ymax=150,
ylabel style={at={(axis description cs:-0.075,.5)},font=\color{white!15!black}},
ylabel={Time [ms]},
axis background/.style={fill=white},
title style={font=\bfseries},
title={\textbf{Average Computation Time per Iteration}},
xmajorgrids,
ymajorgrids,
legend style={at={(0.03,0.97)}, anchor=north west, legend cell align=left, align=left, draw=white!15!black},
title style={font=\Huge},xlabel style={font={\Huge}},ylabel style={font=\Huge},legend style={font=\Huge},tick label style={font=\LARGE}
]
\addplot[ybar, bar width=1, fill=green, draw=black, area legend] table[row sep=crcr] {%
3	14.5891\\
6	41.1476\\
9	82.8050\\
12	143.5704\\
};
\addlegendentry{alg 1}

\addplot[ybar, bar width=1, fill=yellow, postaction={
        pattern=north east lines
    }, draw=black, area legend] table[row sep=crcr] {%
3	6.0073\\
6	6.5142\\
9	6.6539\\
12	6.8428\\
};
\addlegendentry{alg 2}

\end{axis}
\end{tikzpicture}%
}
    \end{subfigure}
    \begin{subfigure}[t]{0.23\textwidth}
        \centering
        \resizebox{\textwidth}{!}{
%
%
\definecolor{mycolor1}{rgb}{0.00000,0.44700,0.74100}%
\definecolor{mycolor2}{rgb}{0.85000,0.32500,0.09800}%
\begin{tikzpicture}

\begin{axis}[%
width=5in,
height=2.5in,
at={(0.758in,0.509in)},
scale only axis,
bar shift auto,
xmin=1.5,
xmax=13.5,
xtick={3, 6, 9, 12},
xlabel style={font=\color{white!15!black}},
xlabel={Number of UAVs},
ymin=0,
ymax=180,
ylabel style={font=\color{white!15!black}},
ylabel={Time [s]},
axis background/.style={fill=white},
title style={font=\bfseries},
title={\textbf{Total Computation Time}},
xmajorgrids,
ymajorgrids,
legend style={at={(0.03,0.97)}, anchor=north west, legend cell align=left, align=left, draw=white!15!black},
title style={font=\Huge},xlabel style={font={\Huge}},ylabel style={font=\Huge},legend style={font=\Huge},tick label style={font=\LARGE}
]
\addplot[ybar, bar width=1, fill=green, draw=black, area legend] table[row sep=crcr] {%
3	3.1804\\
6	30.9430\\
9	84.4611\\
12	176.0173 \\
};
\addlegendentry{alg 1}

\addplot[ybar, bar width=1, fill=yellow, postaction={
        pattern=north east lines
    }, draw=black, area legend] table[row sep=crcr] {%
3	2.5832\\
6	11.3738\\
9	14.2128\\
12	18.8862\\
};
\addlegendentry{alg 2}

\end{axis}
\end{tikzpicture}%
}
    \end{subfigure}
    \caption{Computation time of \Cref{alg:01,alg:02}.} \label{fig:Time Results}
    \vspace*{-1.25\baselineskip} 
\end{figure}
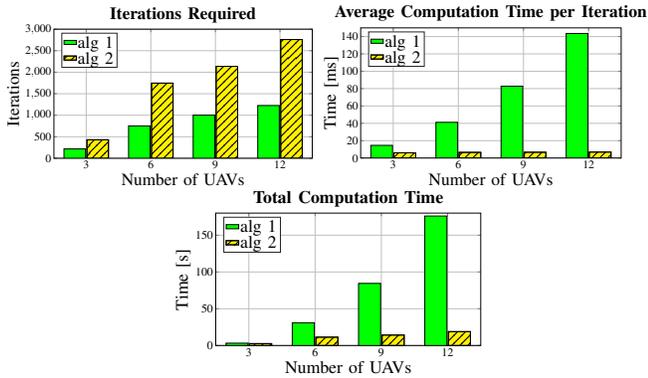

\cref{fig:Time Results} shows the computation time of the results. Additional tests with one, two, and three subgroups were conducted to compare the computation time trends of both algorithms. In this example, only the calculation time of the projection operator was measured, since it significantly exceeds that of other operations.

For this example, the centralized semidefinite program is faster than \cref{alg:01,alg:02}. However, it forces agents to share dynamics information. \cref{alg:01,alg:02} offer a way to solve \cref{opt:Main Problem,opt:Main Chordal Problem} without sharing their dynamics, which is critical for some networked systems.

For smaller systems, the time efficiency of using chordal decomposition does not stand out. However, for larger systems, it can dramatically reduce total computation time, despite the increase in iterations for \cref{alg:02} compared to \cref{alg:01}. This efficiency stems from the parallel projections of \cref{eq:X Chordal proj,eq:Z Chordal proj,eq:W Chordal proj}, and \cref{eq:T Chordal update,eq:U Chordal update,eq:V Chordal update} in \cref{alg:02}. Moreover, the calculation time for \cref{eq:Z proj} is much longer than for \cref{eq:Z Chordal proj} or \cref{eq:W Chordal proj} due to its larger size.

The average computation time for \cref{alg:02} remains relatively consistent across different subgroup numbers, unlike \cref{alg:01}, which shows a proportional increase with the number of subgroups. This consistency in \cref{alg:02} is because chordal decomposition produces the same size \glspl{lmi} in this example, enabling parallel projection to maintain nearly identical average computation time across different subgroup configurations.

\section{CONCLUSIONS} \label{chap:Conclusion}

This paper presents a method analyzing multi-agent system stability without sharing each agent's dynamics. The proposed approach allows each agent to independently analyze its dissipativity while maintaining overall system stability. In addition, the second algorithm significantly reduces the computational burden. The feasibility of these algorithms was demonstrated using a 2D swarm \gls{uav} system. The results show that all approaches successfully proved system stability without sharing its dynamics matrices. Moreover, the efficiency of combining chordal decomposition with \gls{admm} was shown for larger-scale networks. Future work will focus on expanding the distributed stability analysis to enable distributed controller syntheses, ensuring privacy by preventing the exchange of agent's private information.

\addtolength{\textheight}{-12cm}   







\bibliographystyle{IEEEtran}
\bibliography{MyBib}{}

\end{document}